\documentclass[11pt]{article}

\usepackage{mystyle}
\usepackage{multirow}
\usepackage{multicol}

\newcommand{\dist}{{\mbox{dist}}}
\usepackage{microtype}

\title{Super-Fast 3-Ruling Sets\footnote{Part of this work was done while the first author was visiting the University of Iowa as an Indo-US Science and Technology Forum
Research Fellow. The work of the second author is supported in part by 
National Science Foundation grant CCF 0915543.}}
\author{Kishore Kothapalli\footnote{International Institute of Information Technology,  Hyderabad, India 500 032
  \texttt{kkishore@iiit.ac.in}}  $\;\; $and $\;\;$ Sriram Pemmaraju\footnote{Department of Computer Science, The University of Iowa, 
  Iowa City, IA 52242-1419, USA, 
 \texttt{sriram-pemmaraju@uiowa.edu}}
}
\date{}

\begin{document}

\maketitle

\begin{abstract}
A $t$-ruling set of a graph $G = (V, E)$ is a vertex-subset $S \subseteq V$ that is independent
and satisfies the property that every vertex $v \in V$ is at a distance of at most $t$ from some vertex in $S$.
A \textit{maximal independent set (MIS)} is a 1-ruling set.
The problem of computing an MIS on a network is a fundamental problem in distributed algorithms and the fastest
algorithm for this problem is the $O(\log n)$-round algorithm due to Luby (SICOMP 1986) and Alon et al.~(J. Algorithms 1986) from more than 25 years ago.
Since then the problem has resisted all efforts to yield to a sub-logarithmic algorithm.
There has been recent progress on this problem, most importantly an $O(\log \Delta \cdot \sqrt{\log n})$-round 
algorithm on graphs with $n$ vertices and maximum degree $\Delta$, due to Barenboim et al. (Barenboim, Elkin, Pettie, and Schneider, April 2012, arxiv 1202.1983;
to appear FOCS 2012).
The time complexity of this algorithm is sub-logarithmic for $\Delta =
2^{o(\sqrt{\log n})}$. 

We approach the MIS problem from a different angle and ask if $O(1)$-ruling sets
can be computed much more efficiently than an MIS?
As an answer to this question, we show how to compute a 2-ruling
set of an $n$-vertex graph in $O((\log n)^{3/4})$ rounds.
We also show that the above result can be improved for special
classes of graphs. For instance, on high girth graphs (girth 6 or more), trees,
and graphs of bounded arboricity, we show how to compute 3-ruling sets in 
$\exp(O({\sqrt{\log\log n}}))$ rounds, $O((\log \log n)^2 \cdot \log\log\log n)$ rounds,
and $O((\log\log n)^3)$ rounds, respectively.

Our main technique involves randomized sparsification that rapidly reduces the graph degree
while ensuring that every deleted vertex is close to some vertex that remains.
This technique may have further applications in other contexts, e.g., in designing 
sub-logarithmic distributed approximation algorithms.
Our results raise intriguing questions about how quickly an MIS (or 1-ruling sets) 
can be computed, given that 2-ruling sets can be computed in sub-logarithmic rounds.
\end{abstract}

\section{Introduction}
\label{section:Introduction}
\ignore{
1 MIS problem: the fact that it is so fundamental to distributed computing. Why
so. Also, mention Luby’s algorithm and finally the new result due to Barenboim
et al. Also, mention faster algorithms  for special classes of graphs.
}

Symmetry breaking is a fundamental theme in distributed computing and a
classic example of symmetry breaking arises in the computation of a \textit{maximal independent set} (MIS) 
of a given graph. 
About 25 years ago Alon et al.~\cite{AlonBabaiItai} and Luby \cite{LubySICOMP86} independently
devised randomized algorithms for the MIS problem, running in $O(\log n)$ communication rounds.
Since then, all attempts to devise an algorithm for MIS that runs in \textit{sub-logarithmic} rounds 
(for general graphs) have failed. 
Recently, Kuhn et al.~\cite{KuhnMoscibrodaWattenhoferMAIN} proved that there exist
$n$-vertex graphs for which any distributed algorithm, even randomized, that solves the MIS problem requires $\Omega(\sqrt{\log n})$
communication rounds.
Closing this gap between the $O(\log n)$ upper bound and the $\Omega(\sqrt{\log n})$ lower
bound is one of the fundamental challenges in distributed computing.

There has been some exciting recent progress in closing this gap.
Barenboim et al.~\cite{BEPS12} present an algorithm that runs in $O(\log \Delta \sqrt{\log n})$ rounds
on $n$-vertex graphs with maximum degree $\Delta$.
This is sub-logarithmic for  $\Delta \in 2^{o(\sqrt{\log n})}$.
This result uses techniques developed in a paper by Kothapalli et al.~\cite{KOSS06}
for deriving an $O(\sqrt{\log n})$-round algorithm for computing an $O(\Delta)$-coloring
of a $n$-vertex graph with maximum degree $\Delta$.
Barenboim et al.~\cite{BEPS12} also present an algorithm for computing an MIS on trees
in $O(\sqrt{\log n \log\log n})$ rounds.
This is a small improvement over an algorithm from PODC 2011 for computing an MIS on trees due to
Lenzen and Wattenhofer \cite{LenzenWattenhofer} that runs in $O(\sqrt{\log n} \cdot \log\log n)$ rounds.
Barenboim et al.~extend their result on MIS on trees to graphs with girth at least 6 and to graphs
with bounded arboricity.


A problem closely related to MIS, that also involves symmetry breaking at its core,
is the problem of computing $t$-ruling sets. A \textit{$t$-ruling set} of
a graph $G = (V,E)$ is an independent subset $S$ of vertices with the property
that every vertex $v \in V$ is at a distance of at most $t$ from some vertex in
$S$. Thus an MIS is a 1-ruling set\footnote{In the definition of Gfeller and Vicari \cite{GfellerPODC07}, a $t$-ruling set need not be
independent, and what we call a $t$-ruling set, they call an \textit{independent} $t$-ruling set.}. 
In this paper we investigate the distributed complexity of the problem of
computing $t$-ruling sets for $t = O(1)$ with the aim of determining
whether an $O(1)$-ruling set can be computed more efficiently than
an MIS.
For general graphs and for various graph subclasses we show that it is indeed
possible to compute $t$-ruling sets, for small constant $t$, in time that
is much smaller than the best running time for a corresponding MIS
algorithm.
In our first result, we present an algorithm that computes a 2-ruling set in $O((\log n)^{3/4})$ rounds
on general graphs.
Thus we have a sub-logarithmic algorithm for a seemingly minor ``relaxation''
of the MIS problem.
We improve on this result substantially for trees, graphs of girth at least 6,
and graphs of bounded arboricity. For all these subclasses, we present algorithms for
computing 3-ruling sets whose runtime (in rounds) is exponentially faster than the
fastest corresponding MIS algorithms.
For example, for trees our algorithm computes a 3-ruling set in $O((\log \log n)^2 \cdot \log\log\log n)$
communication rounds, whereas the fastest algorithm for MIS on trees takes $O(\sqrt{\log n \log\log n})$
rounds \cite{BEPS12}.

Our work raises intriguing questions on the possibility of faster MIS algorithms
and on the separation between the distributed complexity of $O(1)$-ruling sets and MIS. 
For example, could we design algorithms for MIS that first compute a 2- or 3-ruling set
and then quickly convert that subset to a 1-ruling set?
Is it possible that there are MIS algorithms for trees and related graph subclasses
that run in $O(\mbox{poly}(\log \log n))$ rounds?
Alternately, could the MIS problem be strictly harder than the problem of computing
a $t$-ruling set for some small constant $t$?


Our results should also be viewed in the context of results by Gfeller and Vicari 
\cite{GfellerPODC07}. These authors showed how to compute in $O(\log\log n)$ rounds a vertex-subset $T$ of a given $n$-vertex 
graph $G = (V, E)$ such that (i) every vertex is at most $O(\log \log n)$ hops from some
vertex in $T$ and (ii) the subgraph induced by $T$ has maximum degree $O(\log^5 n)$.
One can use the Barenboim et al.~$O(\log \Delta \sqrt{\log n})$-round MIS algorithm
on $G[T]$ and sparsify $T$ into an $O(\log\log n)$-ruling set in an additional $O(\sqrt{\log n} \cdot \log\log n)$ 
rounds.
Thus, by combining the Gfeller-Vicari algorithm with the Barenboim et al.~algorithm one can compute
an $O(\log\log n)$-ruling set in general graphs in $O(\sqrt{\log n} \cdot \log\log n)$ rounds.
Our result can be viewed as extending the Gfeller-Vicari result by using $t = O(1)$ instead
of $t = O(\log\log n)$.
Also worth noting is the fact that Gfeller and Vicari use their $O(\log\log n)$-ruling set
computation as an intermediate step to computing an MIS on \textit{growth-bounded graphs}.
While the techniques that work for growth-bounded graphs do not work for general graphs or
for the other graph subclasses we consider, this suggests the possibility of getting to an MIS
via a $t$-ruling set for small $t$.

Our technique involves a rapid sparsification of the graph while ensuring that nodes that 
are removed from further consideration are close (within one or two hops) to some remaining node. 
Using this technique we show how to reduce the degrees of graphs rapidly and after sufficiently
reducing the degrees, we can apply MIS algorithms due to Barenboim et al.~\cite{BEPS12}
that take advantage of the low maximum degree.
For example, given a graph $G = (V, E)$ and a parameter $\epsilon$, $0 < \epsilon < 1$, 
our sparsification procedure can run in $O\left( \frac{\log
\Delta}{(\log n)^\epsilon}\right)$ rounds
and partition $V$ into subsets $M$ and $W$ such that with high probability (i) $G[M]$ has maximum
degree $O(2^{(\log n)^\epsilon})$ and (ii) every vertex in $W$ has a neighbor in $M$.
At this stage, we can apply the MIS algorithm of Barenboim et al.~\cite{BEPS12} that runs
in $O(\log\Delta \cdot \sqrt{\log n})$ rounds on $G[M]$.
Since $\Delta(G[M]) = O(2^{(\log n)^\epsilon})$, this step takes $O((\log
n)^{1/2 + \epsilon})$ rounds,
leading to a 2-ruling set algorithm that runs in $O\left(\frac{\log \Delta}{(\log n)^\epsilon} + 
(\log n)^{1/2 + \epsilon}\right)$ rounds.
Picking $\epsilon = 1/4$ yields the $O((\log n)^{3/4})$ rounds 2-ruling set algorithm mentioned above.
We use a similar rapid sparsification approach to derive faster ruling set algorithms for 
different graph subclasses.
We believe that the sparsification technique may be of
independent interest in itself, especially in designing distributed approximation
algorithms that run in sub-logarithmic rounds. 


\subsection{Model}

We consider distributed systems that can be modeled by a graph $G= (V,E)$ with
the vertices representing the computational entities and the edges 
representing communication links between pairs of computational entities. 
We use the standard synchronous, message passing model of communication in which
each node, in each round, can send a possibly distinct message along each incident
edge.
All of our algorithms are structured as a series of ``sparsification'' steps
interleaved with calls to subroutines implementing MIS algorithms on low degree
graphs, due to Barenboim et al.~\cite{BEPS12}.
During the sparsification steps, each node only needs to inform its neighbors
of its membership in some set and therefore each node only needs to send 
the same single bit to all of its neighbors.
Therefore, communication during the sparsification steps can be viewed as occuring in
in a fairly restrictive communication model in which each node is only allowed to 
(locally) broadcast a single bit to all neighbors.
However, some of the MIS algorithms in Barenboim et al. \cite{BEPS12} run
in the $\mathcal{LOCAL}$ model, which allows each node to send a message
of arbitrary size to each neighbor in each round.
Thus, due to their dependency on the MIS algorithms of Barenboim et al.~\cite{BEPS12},
the algorithms in this paper also require the use of the $\mathcal{LOCAL}$ model.

\subsection{Definitions and Notation}

Given a graph $G = (V,E)$, we denote by $N(v)$ the neighborhood of $v$
and by $\deg_G(v)$ the quantity $|N(v)|$. Let $\dist_G(u,v)$ refer to the shortest
distance between any two vertices $u$ and $v$ in $G$. For a subset of vertices
$V' \subseteq V$, let $G[V']$ be the subgraph induced by the
subset $V'$. 

Our calculations make use of Chernoff bounds for tail inequalities on the sum of
independent random variables.
In particular, let $X  := \sum_{i=1}^n\; X_i$ with $E[X_i] = p$ for each
$1\leq i\leq n$. 
The upper tail version of Chernoff bounds that we utilize is:
$\Pr[X \geq E[X]\cdot(1+\epsilon)] \leq \exp(-E[X]\epsilon^2/3)$ for
any $0 < \epsilon < 1$.

In our work, we derive a 3-ruling set algorithm for graphs with bounded
arboricity.
Let the \textit{density} of a graph $G = (V,E)$,
$|V|\ge 2$, be the ratio $\lceil |E|/(|V| - 1) \rceil$. 
Let the density of a single-vertex graph be 1. 
The \textit{arboricity} of a graph
$G = (V, E)$, denoted $a(G)$, can be defined as $a(G) := \max\{ density(G') \mid G' \mbox{ is a subgraph of } G\}$. 
By the celebrated Nash-Williams decomposition theorem \cite{NW64}, the arboricity of
a graph is exactly equal to the minimum number of forests
that its edge set can be decomposed into. 
For examples, trees have arboricity one.
The family of graphs with arboricity $a(G) = O(1)$  includes all planar
graphs, graphs with treewidth bounded
by a constant, graphs with genus bounded by a constant, and the family of graphs that exclude a
fixed minor.
A property of graphs with arboricity $a(G)$ that has
been found useful in distributed computing \cite{BE08,BE09,BE10} is
that the edges of such graphs can be oriented so that each node has at most
$a(G)$ incident edges oriented away from it. 
However, finding such an orientation takes $\Omega(\log n)$ time
\cite{BE08} and since
we are interested in sub-logarithmic algorithms, we cannot rely on the
availability of such an orientation.

\subsection{Our Results}

Here we summarize the results in this paper.
\begin{enumerate}
\item An algorithm, that with high probability, computes a 2-ruling set on general 
graphs in $O\left(\frac{\log \Delta}{(\log n)^\epsilon} + (\log n)^{1/2 + \epsilon}\right)$ rounds for
any $0 < \epsilon < 1$.
Substituting $\epsilon = 1/4$ into this running time expression simplifies it to $O((\log n)^{3/4})$.
\item  An algorithm, that with high probability, computes a 3-ruling 
set on graphs of girth at least 6 
in  $\exp(O(\sqrt{\log\log n}))$ rounds.
\item An algorithm, that with high probability, computes a 3-ruling set
in $O((\log\log n)^2 \log\log\log n)$ rounds on trees.
\item An algorithm, that with high probability, computes a 3-ruling set
on graphs of bounded arboricity in $O((\log\log n)^3)$ rounds. 
\end{enumerate}

Note that all our results run significantly faster than corresponding
algorithms for MIS. In fact, for trees and graphs of bounded arboricity, our
results improve the corresponding results exponentially. 
This is illustrated further in Table \ref{tab:results}. 

\begin{table}[h]
\centering
\begin{tabular}{|l|l|l|l|}
\hline
Graph Class
& {{\bf MIS \cite{BEPS12}}} & {{\bf $O(\log\log n)$-ruling}}
& {\bf 3-ruling set}\\
& & \textbf{sets \cite{GfellerPODC07}} & [This Paper] \\
\hline
\hline
General &  $O(\log \Delta \cdot \sqrt{\log n})$ 
& $O(\sqrt{\log n}\cdot \log\log n)$ & $O((\log n)^{3/4})$  \\
\hline
Trees &    $\tilde{O}(\sqrt{\log n})$
&   & $\tilde{O}((\log\log n)^2)$  \\
\cline{1-2}
\cline{4-4}
Girth $\geq 6$ &   $O(\log\Delta \log\log n + \exp(O(\sqrt{\log\log n})))$ &
& $\exp(O(\sqrt{\log\log n}))$ \\
\cline{1-2}
\cline{4-4}
Bounded & $O(\log \Delta (\log \Delta + \frac{ \log\log n}{\log\log\log n}))$ &  & $O((\log\log n)^3)$ \\
arboricity   &  & & \\
($a = O(1)$) & & & \\
\cline{1-2}
\cline{4-4}
\hline
\end{tabular}
\caption{Comparison of the best known runtimes of distributed algorithms for
MIS, $O(\log\log n)$-ruling sets, and 3-ruling sets.
It should be noted that the algorithm for general graphs described in this paper
computes a 2-ruling set.
Also, we use the notation $\tilde{O}(f(n))$ as a short form for
$O(f(n)\cdot \mbox{polylog}(f(n)))$.
}
\label{tab:results}
\end{table}

\subsection{Related Work}

The work most closely related to ours, which includes the recent work of
Barenboim et al.~\cite{BEPS12} and the work of Gfeller and Vicari \cite{GfellerPODC07}, 
has already been reviewed earlier in this section.

Other work on the MIS problem that is worth mentioning is the elegant MIS
algorithm of M\'{e}tivier et al.~\cite{Metivieretal}.
In this algorithm, each vertex picks a real uniformly at random from the
interval $[0, 1]$ and joins the MIS if its chosen value is a local maxima.
This can be viewed as a variant of Luby's algorithm \cite{LubySICOMP86} and 
like Luby's algorithm, runs in $O(\log n)$ rounds.
Due to its simplicity, this MIS algorithm is used in part by the MIS algorithm
on trees by Lenzen and Wattenhofer \cite{LenzenWattenhofer} and also by
Barenboim et al.~\cite{BEPS12}.

The MIS problem on the class of growth-bounded graphs has attracted fair bit of attention
\cite{KMNWDISC05,GfellerPODC07,WattenhoferSchneiderPODC2008}.
Growth-bounded graphs have the property that the $r$-neighborhood of any vertex $v$
has at most $O(r^c)$ independent vertices in it, for some constant $c > 0$.
In other words, the rate of the growth of independent sets is polynomial in
the radius of the ``ball'' around a vertex.
Schneider and Wattenhofer \cite{WattenhoferSchneiderPODC2008} showed 
that there is a deterministic MIS algorithm on growth-bounded
graphs that runs in $O(\log^* n)$ rounds.
Growth-bounded graphs have been used to model wireless networks because the number of independent
vertices in any spatial region is usually bounded by the area or volume of that region.
In contrast to growth-bounded graphs, the graph subclasses we consider in this paper tend to
have arbitrarily many independent vertices in any neighborhood.

Fast algorithms for $O(1)$-ruling sets may have applications in distributed approximation
algorithms.
For example, in a recent paper by Berns et al.~\cite{BernsHegemanPemmaraju}
a 2-ruling set is computed as a way of obtaining a $O(1)$-factor approximation to the
metric facility location problem.
Our work raises questions about the existence of sub-logarithmic round algorithms
for problems such as minimum dominating set, vertex cover, etc., at least for
special graph classes.

\subsection{Organization of the Paper}
The rest of the paper is organized as follows. Section
\ref{section:generalGraphs} shows our
result for general graphs. Section \ref{section:girth5} shows our results for graphs
of girth at least 6, and for trees.  Section \ref{sec:arboricity} extends the 
results of Section \ref{section:girth5} to graphs of arboricity bounded by a 
poly-logarithmic value. The paper ends with some concluding remarks and open 
problems in Section \ref{sec:conclusions}.

\ignore{
2. Definition of ruling sets. MIS as a 1-ruling set. The question we want to
answer is: can we compute k-ruling sets for small k very fast?

3. Statement of our results. This should include definition of bounded
arboricity graphs.

4. Some discussion on the significance of our results.
}

\section{2-Ruling Sets in General Graphs}
\label{section:generalGraphs}

In this section we describe Algorithm \textsc{RulingSet-GG}, that runs in sub-logarithmic rounds
and computes a 2-ruling set in general graphs.
The reader is encouraged to consult the pseudocode of this algorithm while reading the following text.
Let $f$ be the quantity $2^{(\log n)^\epsilon}$ for some parameter $0 < \epsilon < 1$.
Let $i^*$ be the smallest positive integer such that $f^{i^* + 1} \ge \Delta$.
Thus $i^* = \lceil \log_f \Delta \rceil - 1$.
It is also useful to note that $i^* = O\left(\frac{\log \Delta}{(\log n)^\epsilon}\right)$.
The algorithm proceeds in \textit{stages} and there are $i^*$ stages, indexed by $i = 1, 2, \ldots, i^*$.
In Stage $i$, all ``high degree'' vertices, i.e., vertices with degrees greater than $\frac{\Delta}{f^i}$,
are processed.
Roughly speaking, in each stage we peel off from the ``high degree'' vertex set, a subgraph with degree
bounded above by $O(f \cdot \log n)$.
Following this we also peel off all neighbors of this subgraph.
More precisely, in Stage $i$ each ``high degree'' vertex joins a set $M_i$ with probability
$\frac{6 \log n \cdot f^i}{\Delta}$ (Line 6).
Later we will show (in Lemma \ref{lemma:degreeReductionGG}) that with high probability any vertex that is in $V$ at the start
of Stage $i$ has degree at most $\Delta/f^{i-1}$. (This is trivially true for $i = 1$.)
Therefore, it is easy to see that any vertex in the graph induced by $M_i$ has expected degree at most $O(f \cdot \log n)$.
In fact, this is true with high probability, as shown in Lemma \ref{lemma:degreeBound}.
This degree bound allows the efficient computation of an MIS on the subgraph induced by $M_i$.
Following the identification of the set $M_i$, all neighbors of $M_i$ that are outside 
$M_i$ are placed in a set $W_i$ (Line 9).
Both sets $M_i$ and $W_i$ are then deleted from the vertex set $V$.
The sets $W_i$ play a critical role in our algorithm.
For one, given the probability $\frac{6\log n \cdot f^i}{\Delta}$ of joining 
$M_i$, we can show that with high probability every ``high degree'' vertex 
ends up either in $M_i$ or in $W_i$.
This ensures that all ``high degree'' vertices are deleted from $V$ in each
Stage $i$.
Also, the sets $W_i$ act as ``buffers'' between the $M_i$'s ensuring that
there are no edges between $M_i$ and $M_{i'}$ for $i \not= i'$.
As a result the graph induced by $\cup_i M_i$ also has low degree, i.e.,
$O(f \cdot \log n)$.
Therefore, we can compute an MIS on the graph induced by $\cup_i M_i$ 
in ``one shot'' rather than deal with each of the graphs induced by
$M_1, M_2, \ldots$ one by one.

Given the way in which ``high degree'' vertices disappear from $V$,
at the end of all $i^*$ stages, the graph $G$ induced by vertices
that still remain in $V$ would have shrunk to the point where the maximum
degree of a vertex in $G$ is $O(f)$.
The algorithm ends by computing an MIS on the graph induced by $V \cup
(\cup_i M_i)$.
As mentioned before, the $M_i$'s do not interact with each other or
with $V$ and therefore the degree of the graph induced by $(\cup_i M_i) \cup V$ is $O(f \cdot \log n)$.
We use the MIS algorithm due of Barenboim et al.~\cite{BEPS12} that
runs in $O(\log \Delta \cdot \sqrt{\log n})$ rounds for this purpose.
Since $\Delta = O(f \cdot \log n)$ and $f= 2^{{(\log n)}^\epsilon}$,
this step runs in $O((\log n)^{\frac{1}{2} + \epsilon})$ rounds.
In the algorithm described below, we denote by \texttt{MIS-LOWDEG} the subroutine that implements the Barenboim et al.~algorithm.
We use $H$ to denote a static copy of the input graph $G$.
\begin{tabbing}
tabtab \= tab \= tab  \= tab \= tab \= \kill
\textbf{Algorithm} \textsc{RulingSet-GG}$(G = (V, E))$\\
1.\> $f \leftarrow 2^{(\log n)^\epsilon}$; $H \leftarrow G$\\
2.\> \textbf{for} $i \leftarrow 1,2,\ldots,i^*$ \textbf{do}\\
\>\>   \textbf{/* Stage $i$ */}\\
3.\>\> $M_i \leftarrow \emptyset$; $W_i \leftarrow \emptyset$;\\
4.\>\> \textbf{for} each $v \in V$ \textit{in parallel} \textbf{do}\\
5.\>\>\> \textbf{if} $\deg_G(v) > \frac{\Delta}{f^i}$ \textbf{then}\\
6.\>\>\>\> $M_i \leftarrow M_i \cup \{v\}$ with probability $\frac{6\log n \cdot f^i}{\Delta}$\\
7.\>\> \textbf{for} each $v \in V$ \textit{in parallel} \textbf{do}\\
8.\>\>\> \textbf{if} $v \in N(M_i) \setminus M_i$ \textbf{then}\\
9.\>\>\>\> $W_i \leftarrow W_i \cup \{v\}$\\ 
10.\>\> $V \leftarrow V \setminus (M_i \cup W_i)$\\
\> \textbf{end-for}($i$) \\
11.\>$I \leftarrow \texttt{MIS-LOWDEG}(H[(\cup_i M_i) \cup V])$\\
\>\textbf{return} $I$; \\
\end{tabbing}

\begin{lemma}
\label{lemma:degreeReductionGG}
At the end of Stage $i$, $1 \le i \le i^*$, with probability at least $1 - \frac{1}{n^{5}}$ all vertices still in $V$ have degree at most $\frac{\Delta}{f^i}$.
\end{lemma}
\begin{proof}
Consider a ``high degree'' vertex $v$, i.e., a vertex with degree more than $\Delta/f^i$, at the start of Stage $i$.
Then,
\begin{eqnarray*}
\Pr[v\mbox{ is added to }M_i \cup W_i] & \ge & 1 - \left(1 - \frac{6 \log n \cdot f^i}{\Delta}\right)^{\frac{\Delta}{f^i}}\\
                                              & \ge & 1 - e^{-6 \cdot \log n} = 1 - \frac{1}{n^6}
\end{eqnarray*}
Therefore, using the union bound, we see that with probability at least $1 - \frac{1}{n^{5}}$ all vertices in $V$ 
that have degree more than $\Delta/f^i$ at the start of Stage $i$ will join $M_i \cup W_i$ in Stage $i$.
\end{proof}

\begin{lemma}
\label{lemma:degreeBound}
Consider a Stage $i$, $1 \le i \le i^*$.
With probability at least $1 - \frac{2}{n}$, the subgraph induced by $M_{i}$ (i.e., $H[M_i]$) has maximum degree $12\log n \cdot f$.
\end{lemma}
\begin{proof}
We condition on the event that all vertices that are in $V$ at the beginning
of Stage $i$ have degree at most $\frac{\Delta}{f^{i-1}}$.
For $i = 1$, this event happens with probability 1 and for $i > 1$,
Lemma \ref{lemma:degreeReductionGG} implies that this event happens
with probability at least $1 - 1/n^{5}$.
Consider a vertex $v \in V$ that is added to $M_i$.
Let $\deg_{M_i}(v)$ denote the degree of vertex $v$ in $H[M_i]$.
Then,
$$E[\deg_{M_i}(v)] \le \frac{\Delta}{f^{i-1}} \cdot \frac{6\log n \cdot f^i}{\Delta} = 6 \log n \cdot f.$$
Here we use the fact that $\deg_G(v) \le \frac{\Delta}{f^{i-1}}$ for all $v \in V$ at the start of Stage $i$.
Since vertices join $M_i$ independently, using Chernoff bounds we conclude that
$\Pr[\deg_{M_i}(v) \ge 12 \log n \cdot f] \le 1/n^2$.
Therefore, with probability at least $1 - 1/n$ the maximum degree of $H[M_i]$
is at most $12 \log n \cdot f$.
We now drop the conditioning on the event that all vertices that are in $V$ at
the beginning of Stage $i$ have degree at most $\frac{\Delta}{f^{i-1}}$ and use
Lemma \ref{lemma:degreeReductionGG} and the union bound to obtain the lemma.
\end{proof}

\begin{theorem}
\label{theorem:finalGG}
Algorithm \textsc{RulingSet-GG} computes a 2-ruling set of the input graph $G$ in
$O(\frac{\log \Delta}{(\log n)^\epsilon} + (\log n)^{1/2 + \epsilon})$
rounds.
\end{theorem}
\begin{proof}
It is easy to see that every stage of the algorithm runs in $O(1)$ 
communication rounds.
Since there are $i^*$ stages and since $i^* = O\left(\frac{\log \Delta}{(\log n)^\epsilon}\right)$, the running time of the stages all together is
$O\left(\frac{\log \Delta}{(\log n)^\epsilon}\right)$.
From Lemma \ref{lemma:degreeReductionGG} we see that the vertex set $V$ remaining after all $i^*$ stages induces a graph with maximum degree $f$
with high probability.
From Lemma \ref{lemma:degreeBound} we see that the maximum degree of every $H[M_i]$ is bounded above by $O(f \cdot \log n)$ with high probability.
Furthermore, since there is no interaction between any pair of $M_i$'s and also between $V$ and the $M_i$'s,
the maximum degree of the graph induced by $(\cup_i M_i) \cup V$ is also $O(f \cdot \log n)$.
Therefore, with high probability, the MIS computation at the end of the algorithm takes 
$O((\log n)^{1/2+\epsilon})$ rounds using \cite[Theorem 4.3]{BEPS12}.
Together these observations yield the claimed running time.

To see that $I$ is a 2-ruling set, first observe that every vertex $v$
ends up in $M_i \cup W_i$ for some $1 \le i \le i^*$ or remains in $V$ until the end.
If $v$ ends up in $W_i$, it is at most 2 hops from a vertex in $I$ that
belongs to the MIS of $H[M_i]$.
Otherwise, $v$ is 1 hop away from a vertex in $I$.
\end{proof}

Using $\epsilon = 1/4$ in the above theorem results in 
Corollary \ref{cor:simple}. A further optimization on the choice of
$\epsilon$ for graphs with degree in $2^{\omega(\sqrt{\log n}})$ is
shown in Corollary \ref{corollary:finalGG}.

\begin{corollary}
\label{cor:simple}
Algorithm \textsc{RulingSet-GG} computes a 2-ruling set of the input
graph $G$ in $O((\log n)^{3/4})$ rounds.
\end{corollary}

\begin{corollary}
\label{corollary:finalGG}
(i) For a graph $G$ with $\Delta = 2^{O(\sqrt{\log n})}$, Algorithm \textsc{RulingSet-GG} computes a 2-ruling set of the input graph $G$ in
$O((\log n)^{1/2 + \epsilon})$ rounds for any $\epsilon > 0$.
(i) For a graph $G$ with $\Delta = 2^{\omega(\sqrt{\log n})}$, Algorithm \textsc{RulingSet-GG} computes a 2-ruling set of the input graph $G$ in
$O((\log n)^{1/4}\sqrt{\log \Delta})$ rounds.
\end{corollary}
\begin{proof}
We get (i) by simply plugging $\Delta = 2^{O(\sqrt{\log n})}$ into the running time expression from Theorem \ref{theorem:finalGG}.
(ii) In this case, we know that $\log \Delta = \omega(\sqrt{\log n})$ and $\log \Delta \le \log n$.
Consider the two expressions $\frac{\log \Delta}{(\log n)^\epsilon}$ and $(\log n)^{1/2 + \epsilon}$ in the running time
expression from Theorem \ref{theorem:finalGG}.
At $\epsilon = 0$ the first term is larger and as we increase $\epsilon$, the first term falls and
the second term increases.  
By the time $\epsilon = 1/4$ the second term is larger.
We find a minimum value by equating the two terms and solving for $\epsilon$.
This yields an ``optimal'' value of 
$$\epsilon = \frac{\log \log \Delta}{2 \log \log n} - \frac{1}{4}$$
and plugging this into the running time expression yields the running time bound of $O((\log n)^{1/4} \cdot \sqrt{\log \Delta})$ rounds.
\end{proof}

\section{3-Ruling Sets for High Girth Graphs and Trees}
\label{section:girth5}

Our goal in this section is to devise an $O(1)$-ruling set algorithm for high girth graphs
and trees that is much faster than the 2-ruling set algorithm for general graphs from the
previous section.
In Algorithm \textsc{RulingSet-GG} we allow the graph induced by $M_i$ to have
degree as high as $O(f \cdot \log n)$ where $f = 2^{(\log n)^{\epsilon}}$. 
Computing an MIS on a graph with degree as high as this is too time consuming for our purposes.
We could try to reduce $f$, but this will result in a corresponding increase in the number of 
stages.
Therefore, we need to use additional ideas to help simultaneously keep
the maximum degree of the graphs $H[\cup_i M_i]$ small and also the number
of stages small.



Let $G = (V, E)$ be a graph with $n$ vertices, maximum degree $\Delta$, and girth 
at least 6.  Let $i^*$ be the smallest positive integer such that
$\Delta^{1/2^{i^*}} \le  6 \cdot \log n$.  It is easy to check that $i^* =
O(\log \log \Delta)$.

Let $M_1$ and $M_2$ be disjoints subsets of $V$ such that the maximum vertex degree in
$G[M_1]$ and in $G[M_2]$ is bounded by $O(\log n)$.
We use $\texttt{MIS-TWOSTAGE}(G, M_1, M_2)$ to denote a call to the following algorithm for computing
an MIS on $G[M_1 \cup M_2]$.
\begin{enumerate}
\item Compute an MIS $I_1$ on $G[M_1]$ using the algorithm of Barenboim et al. (\cite{BEPS12}, Theorem 7.2).
\item Compute an MIS $I_2$ on $G[M_2 \setminus N(I_1)]$ using the algorithm of Barenboim et al. (\cite{BEPS12}, Theorem 7.2).
\item return $I_1 \cup I_2$.
\end{enumerate}
This algorithm runs in $\exp(O(\sqrt{\log\log n}))$ rounds since
the maximum degree in $G[M_1]$ and in $G[M_2]$ is bounded by $O(\log n)$
and therefore by Theorem 7.2 \cite{BEPS12} each of the MIS computations requires 
$\exp(O(\sqrt{\log\log n}))$ rounds.
If $G$ were a tree, then we could use Theorem 7.3 in Barenboim et al. \cite{BEPS12},
which tells us that we can compute an MIS on a tree with maximum degree $O(\log n)$ in
$O(\log\log n \cdot \log\log\log n)$ rounds.
From this we see that a call to $\texttt{MIS-TWOSTAGE}(G, M_1, M_2)$ runs in $O(\log\log n \cdot \log\log\log n)$
rounds when $G$ is a tree.

In our previous algorithm, Algorithm \textsc{RulingSet-GG}, we used
degree ranges $(\frac{\Delta}{f}, \Delta]$, $(\frac{\Delta}{f^2}, \frac{\Delta}{f}]$, etc.
Here we use even larger degree ranges: $(\Delta^{1/2}, \Delta]$, $(\Delta^{1/4}, \Delta^{1/2}]$, etc.
The algorithm proceeds in stages and in  Stage $i$ all vertices with degrees in the range $(\Delta^{1/2^i}, \Delta^{1/2^{i-1}}]$
are processed.
To understand the algorithm and why it works consider what happens in Stage
1. (It may be helpful to consult the pseudocode of Algorithm
\textsc{RulingSet-HG} while reading the following.)
In Line 6 we allow ``high degree'' vertices (i.e., those with degree more than $\sqrt{\Delta}$) 
to join a set $M_1$ with a probability $\frac{6\log n}{\Delta}$. 
This probability is small enough that it ensures that the expected maximum
degree of the subgraph induced by $M_1$ is $O(\log n)$.
In fact, this also holds with high probability, as shown in Lemma \ref{lemma:Mdegree}. 
However, as can be seen easily, there are lots of ``high degree'' vertices
that have no neighbor in $M_1$.
We use two ideas to remedy this situation.
The first idea is to allow ``low degree'' vertices (i.e., those with degree
at most $\sqrt{\Delta}$) also to join a set $M_2$, with the somewhat
higher probability of $\frac{6 \log n}{\sqrt{\Delta}}$ (Line 7).
This probability is low enough to ensure that the graph induced by
$M_2$ has $O(\log n)$ maximum degree, but it is also high enough to
ensure that if a ``high degree'' node has lots of ``low degree'' neighbors, 
it will see some neighbor in $M_2$, with high probability.
This still leaves untouched ``high degree'' vertices with lots of 
``high degree'' neighbors. 
To deal with these vertices, we remove not just the neighborhood of
$M_1$, but also the 2-neighborhood of $M_1$.
The fact that $G$ has a high girth ensures that a ``high degree'' vertex
that has many ``high degree'' neighbors has lots of vertices in its
2-neighborhood.
This allows us to show that such ``high degree'' vertices are also removed
with high probability. The above arguments are formalized in Lemma
\ref{lemma:activeDegreeReduction}.
We repeat this procedure for smaller degree ranges
until the degree of the graph that remains is
poly-logarithmic. Figure \ref{figure:mw} shows one iteration of the
algorithm. Pseudocode of our algorithm appears as Algorithm
\textsc{RulingSet-HG} below.

\begin{tabbing}
tabtab \= tab \= tab  \= tab \= tab \= \kill
\textbf{Algorithm} \textsc{RulingSet-HG}$(G = (V, E))$\\
1.\> $I \leftarrow \emptyset$\\
2.\> \textbf{for} $i=1,2,\cdots,i^*$ \textbf{do}\\ 
\> \textbf{/* Stage $i$ */}\\
3.\>\> $M_1 \leftarrow \emptyset$; $M_2 \leftarrow \emptyset$; $W \leftarrow \emptyset$\\
4.\>\> \textbf{for} $v \in V$ \textit{in parallel} \textbf{do}\\
5.\>\>\> \textbf{if} $\deg(v) > \Delta^{1/2^i}$ \textbf{then}\\
6.\>\>\>\> $M_1 \leftarrow M_1 \cup \{v\}$ with probability $\frac{6 \cdot \log n}{\Delta^{1/2^{i-1}}}$\\
\>\>\> \textbf{else if} $\deg(v) \le \Delta^{1/2^i}$ \textbf{then}\\
7. \>\>\>\> $M_2 \leftarrow M_2 \cup \{v\}$ with probability $\frac{6 \cdot \log n}{\Delta^{1/2^i}}$\\
8.\>\> $I \leftarrow I \cup \texttt{MIS-TWOSTAGE}(G, M_1, M_2)$\\
9.\>\> \textbf{for} $v \in V \setminus (M_1 \cup M_2)$ \textit{in parallel} \textbf{do}\\
10.\>\>\> \textbf{if} $\mbox{dist}(v, M_1 \cup M_2) \le 2$ \textbf{then}\\
11.\>\>\>\> $W \leftarrow W \cup \{v\}$\\ 
12.\>\> $V \leftarrow V \setminus (M_1 \cup M_2 \cup W)$\\
\> \textbf{end-for}($i$) \\
13.\>$I \leftarrow I \cup \texttt{MIS}(G)$\\
\>\textbf{return} $I$; \\
\end{tabbing}

\begin{figure}
\centering
\epsfig{file=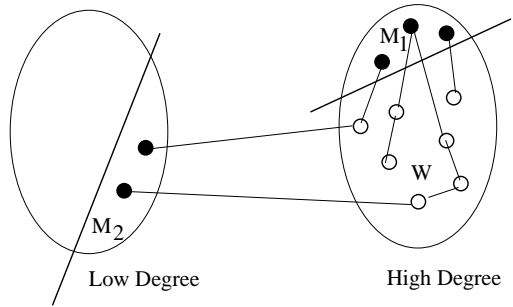,scale=0.6}
\caption{Figure showing one iteration of  Algorithm \textsc{RulingSet-HG}.
The figure shows the sets $M_1$, $M_2$ and $W$. 
}
\label{figure:mw}
\end{figure}

In the following, we analyze Algorithm \textsc{RulingSet-HG}. We show in
Lemma \ref{lemma:activeDegreeReduction} that all nodes of degree at
least $\Delta^{1/2^i}$ can be processed in the $i$th iteration. This is
followed by Lemma \ref{lemma:Mdegree} that argues that the degree of
$G[M_1\cup M_2]$ is $O(\log n)$, and finally Theorem \ref{theorem:treesHighGirth}
that shows our result for graph of girth at least 6 and trees.

\begin{lemma}
\label{lemma:activeDegreeReduction}
For $1 \le i \le i^*$, with probability at least $1 - 1/n^2$, all 
vertices still in $V$ have degree at most $\Delta^{1/2^i}$ at the 
end of iteration $i$.
\end{lemma}
\begin{proof}
Consider a vertex $v \in V$ at the start of iteration $i$ that has 
degree greater than $\Delta^{1/2^i}$.
Vertex $v$ can have one of two types:
\begin{description}
\item[Type I]: $v$ is of Type I if at least half of $v$'s neighbors 
have degree greater than $\Delta^{1/2^i}$.
\item[Type II]: $v$ is of Type II if fewer than half of $v$'s neighbors 
have degree greater than $\Delta^{1/2^i}$.
\end{description}
\vspace{0.1in}
If $v$ is of Type I, then there are at least $1/2 \cdot \Delta^{1/2^i} \cdot
\Delta^{1/2^i} = \Delta^{1/2^{i-1}}/2$ vertices in $v$'s 2-neighborhood.
Here we use the fact that $G$ has girth at least 6.
Now note that any vertex $u$ in $v$'s 2-neighborhood is added to $M_1 \cup
M_2$ with probability at least $\frac{6\log n}{\Delta^{1/2^{i-1}}}$.
Therefore, the probability that no vertex in $v$'s 2-neighborhood is added to
$M_1 \cup M_2$ is at most $(1 - \frac{6\log n}{\Delta^{1/2^{i-1}}})^{|N_2(v)|}$,
where $N_2(v)$ denotes the 2-neighborhood of vertex $v$.
Here we use the fact that vertices are added to $M_1 \cup M_2$ independently.
Using the lower bound $|N_2(v)| \ge \Delta^{1/2^{i-1}}/2$, we see that 
\begin{eqnarray*}
\Pr[v\mbox{ is added to }M_1 \cup M_2 \cup W] & \ge & 1 - \left(1 - \frac{6 \cdot \log n}{\Delta^{1/2^{i-1}}}\right)^{\frac{\Delta^{1/2^{i-1}}}{2}}\\
                                              & \ge & 1 - e^{-3 \cdot \log n} = 1 - \frac{1}{n^3}
\end{eqnarray*}
If $v$ is of Type II, then more than half of $v$'s neighbors have degree less than or equal to $\Delta^{1/2^i}$.
Each such ``low degree'' neighbor is added to $M_2$ with probability
$6 \log n/\Delta^{1/2^i}$.
Therefore,
\begin{eqnarray*}
\Pr[v\mbox{ is added to }M_1 \cup M_2 \cup W] & \ge & 1 - \left(1 - \frac{6 \cdot \log n}{\Delta^{1/2^i}}\right)^{\frac{\Delta^{1/2^i}}{2}}\\
                                              & \ge & 1 - e^{-3 \cdot \log n} 
= 1 - \frac{1}{n^{3}}
\end{eqnarray*}
In either case, $v$ is added to $M_1 \cup M_2 \cup W$ with probability at least $1 - 1/n^{3}$.
Therefore, by the union bound every node of degree greater than $\Delta^{1/2^i}$ is added to $M_1 \cup M_2 \cup W$ 
with probability at least $1 - 1/n^{2}$. 
Therefore, at the end of iteration $i$, with probability at least $1 - 1/n^2$,
there are no vertices in $V$ with degree more than $\Delta^{1/2^i}$.
\end{proof}

\begin{corollary}
With probability at least $1 - 1/n^2$, after all $i^*$ iterations of the for-loop in Algorithm \textsc{RulingSet-HG},
the graph $G$ has maximum degree at most $6\log n$.
\end{corollary}

\begin{lemma}
\label{lemma:Mdegree}
Consider an arbitrary iteration $1 \le i \le i^*$ and let 
$H = G[M_1 \cup M_2]$.
With probability at least $1 - 2/n$, the maximum degree of a vertex in $H[M_j]$, $j = 1, 2$
is at most $12 \cdot \log n$.
\end{lemma}
\begin{proof}
We condition on the event that all vertices that are in $V$ at the beginning 
of an iteration $i$ have degree at most $\Delta^{1/2^{i-1}}$.
For $i = 1$, this event happens with probability 1 and for $i > 1$,
Lemma \ref{lemma:activeDegreeReduction} implies that this event happens
with probability at least $1 - 1/n^2$.
Consider a vertex $v \in V$ that is added to $M_1$. 
Let $\deg_{M_1}(v)$ denote the degree of vertex $v$ in $G[M_1]$.
Then,
$$E[\deg_{M_1}(v)] \le \Delta^{1/2^{i-1}} \cdot \frac{6 \cdot \log n}{\Delta^{1/2^{i-1}}} = 6 \cdot \log n.$$
Here we use the fact that $\deg(v) \le \Delta^{1/2^{i-1}}$ for all $v \in V$ at the start of iteration $i$.
Similarly, for a vertex $v \in V$ that is added to $M_2$, let $\deg_{M_2}(v)$
denote the degree of vertex $v$ in $G[M_2]$. Then,
$$E[\deg_{M_2}(v)] \le \Delta^{1/2^i} \cdot \frac{6 \cdot \log n}{\Delta^{1/2^i}} = 6 \cdot \log n.$$
Here we use the fact that $v$ is added to $M_2$ only if $\deg(v) \le \Delta^{1/2^i}$.
Since vertices join $M_1$ independently, using Chernoff bounds we conclude that
$\Pr[\deg_{M_1}(v) \ge 12 \cdot \log n] \le 1/n^2$.
Similarly, we conclude that 
$\Pr[\deg_{M_2}(v) \ge 12 \cdot \log n] \le 1/n^{2}$.
Therefore, with probability at least $1 - 1/n$ the maximum degree of $G[M_1 \cup
M_2]$ is at most $12 \log n$.
We now drop the conditioning on the event that all vertices that are in $V$ at 
the beginning of iteration $i$ have degree at most $\Delta^{1/2^{i-1}}$ and use
Lemma \ref{lemma:activeDegreeReduction} and the union bound to obtain the lemma.
\end{proof}

\begin{theorem}
\label{theorem:treesHighGirth}
Algorithm \textsc{RulingSet-HG} computes a 3-ruling set of $G$.
If $G$ is a graph with girth at least 6 then \textsc{RulingSet-HG}
terminates in $\exp(O(\sqrt{\log \log n}))$ rounds with high probability.
If $G$ is a tree then \textsc{RulingSet-HG} terminates in $O((\log\log n)^2 \cdot \log\log\log n)$ rounds with high probability.
\end{theorem}
\begin{proof}
Consider a vertex $v \in V$ that is added to $M_1 \cup M_2 \cup W$ in some
iteration $i$.
Since the algorithm computes an MIS on $G[M_1 \cup M_2]$ and since every
vertex in $W$ is at most 2 hops (via edges in $G$) from some vertex in $M_1
\cup M_2$, it follows that $v$ is at distance at most 3 from a vertex placed
in $I$ in iteration $i$.
A vertex that is not added to $M_1 \cup M_2 \cup W$ ends up 
in the graph whose MIS is computed (in Line 13) and is therefore at 
most 1 hop away from a vertex in $I$.
Thus every vertex in $V$ is at most 3 hops away from some vertex in $I$.

The total running time of the algorithm is $i^*$ times the worst case running time 
the call to the \texttt{MIS} subroutine in Line 8 plus the running time of the call
to the \texttt{MIS} subroutine in Line 13.
This implies that in the case of graphs of girth at least 6, Algorithm \textsc{RulingSet-HG} 
runs in $\exp(O(\sqrt{\log\log n})) \cdot O(\log\log \Delta) =
\exp(O(\sqrt{\log\log n}))$  rounds. 
In the case of trees, Algorithm \textsc{RulingSet-HG} runs in 
$O(\log\log \Delta \cdot \log\log n \cdot \log\log \log n) = O((\log\log n)^2
\cdot \log\log\log n)$ rounds.
\end{proof}

\section{Graphs with Bounded Arboricity}
\label{sec:arboricity}
In the previous section, we used the fact that the absence of
short cycles induces enough independence so that in each iteration, with high
probability the ``high degree'' nodes join the set $M_1\cup M_2\cup W$. This has
allowed us to process nodes of degrees in the
range $(\Delta^{1/2^i},\Delta^{1/2^{i-1}}]$ in iteration $i$. 
In this section, we show that a 3-ruling set can be computed even in the
presence of short cycles provided the  graph has an arboricity bounded by
$\log^k n$ for a constant $k$. The algorithm we use for this case is essentially
similar to that of Algorithm {\sc RulingSet-HG} from Section
\ref{section:girth5}. Recall from Section \ref{section:girth5} that $i^*$ refers
to the smallest positive integer such that $\Delta^{1/2^{i^*}} \leq
6\cdot \log n$. We make the following changes to Algorithm
\textsc{RulingSet-HG} to adapt it to graphs of arboricity $a = a(G)$.

\begin{itemize}
\item In iteration $i$, for $1\leq i \leq i^*$, a node $v$ that has a degree 
at least $\Delta^{1/2^i}$ joins the set $M_1$ with probability $\frac{6\cdot a\log
n}{\Delta^{1/2^{i-1}}}$. (See Line 6 of Algorithm \textsc{RulingSet-HG}.)
\item In iteration $i$, for $1\leq i \leq i^*$, a node $v$ with degree less than $\Delta^{1/2^i}$
joins $M_2$ with probability  $\frac{6\cdot a\log n}{\Delta^{1/2^i}}$. (See Line
7 of Algorithm \textsc{RulingSet-HG}).
\end{itemize}

In the following, we show lemmas equivalent to Lemma
\ref{lemma:activeDegreeReduction},\ref{lemma:Mdegree} for a graph with $a \in O(\log^k n)$ for a
constant $k$. 

\ignore{
During the analysis, we use the fact that the edges of graphs of 
arboricity $a$ can be oriented so that each node has at most $a$ outgoing edges.
We note that nodes need not compute such an orientation, but the orientation is
used only for the purposes of the analysis. 
\note{TO ELABORATE}
}

\begin{lemma}
\label{lemma:arb1}
Consider any iteration $i$ for $1\leq i \leq i^*$. With probability at least
$1-\frac{1}{n^{2}}$, all nodes still in $V$ have degree at most
$\Delta^{1/2^i}$ at the end of iteration $i$.
\end{lemma}
\begin{proof}
For $i=0$, we see that each vertex has degree at most $\Delta$ with probability
1. Hence, the lemma holds for $i=0$. Let us assume inductively that the lemma
holds through the first $i-1$ iterations and let us consider the $i$th
iteration.

Consider a node $v$ still in $V$ at the start of iteration $i$ that has degree
at least $\Delta^{1/2^i}$. We distinguish between two cases. Recall that
for a vertex $v$, $N_2(v)$ refers to the 2-neighborhood of $v$.

\begin{itemize}
\item $v$ has at least half its neighbors each with degree at least $\Delta^{1/2^i}$.
In this case, we notice that $v$ has at least
$\Delta^{1/2^{i-1}}/2a$ nodes at a distance of 2 from $v$. Otherwise, the graph
induced by the set $N(v)\cup N_2(v)$ has an arboricity greater than $a$, which
is a contradiction.
Each of the vertices $u \in N_2(v)$ joins $M_1 \cup M_2$ with probability at
least $\frac{6\cdot a\log n}{\Delta^{1/2^{i-1}}}$. Therefore, 
\[
\begin{array}{lll}
\Pr(v \in M_1\cup M_2 \cup W) & \geq & 1 - (1-\frac{6\cdot a\log
n}{\Delta^{1/2^{i-1}}})^{\Delta^{1/2^{i-1}}/2a} \\
& \geq & 1 - e^{6 \log n/2}  =  1 - 1/n^{3}
\end{array}
\]

\item $v$ has at most half its neighbors each with degree at least 
$\Delta^{1/2^{i}}$. In this case, each such neighbor of $v$ joins $M_2$ with
probability $\frac{c\cdot a\log n}{\Delta^{1/2^{i}}}$. Therefore, we can compute
the probability that $v \in M_1\cup M_2\cup W$ as follows.
\[
\begin{array}{lll}
\Pr(v \in M_1\cup M_2 \cup W) & \geq & 1 - (1-\frac{6\cdot a\log
n}{\Delta^{1/2^{i}}})^{\Delta^{1/2^{i}}/2a} \\
& \geq & 1 - e^{6 \log n/2} =  1 - 1/n^{3}
\end{array}
\]
\end{itemize}

In either case we see that $v$ joins $M_1\cup M_2\cup W$ with a probability of
$1/n^{3}$. Using the union bound, as in the proof
of Lemma \ref{lemma:activeDegreeReduction}, vertices still in $V$ have degree at
most $\Delta^{1/2^i}$ with probability at most $1-\frac{1}{n^2}$.
\end{proof}

Lemma \ref{lemma:Mdegree} also holds with the change that the graph $H[M_j]$ for
$j=1,2$ as defined in Lemma \ref{lemma:Mdegree} has a degree at most
$12\cdot a\log n$. Since $a \in O(\log^k n)$, the above degree is in $O(\log^{k+1}
n)$, with high probability. Therefore, the following theorem holds.

\begin{theorem}
Algorithm \textsc{RulingSet-HG} computes a 3-ruling set of a graph $G$ of
arboricity $a \in O(\log^k  n)$, for a constant $k$, in  $O(\sqrt{\log n}\cdot
(\log\log n)^2+\log^{3/4} n \log\log n)$ rounds.
Further, if $a = O(1)$, then Algorithm \textsc{RulingSet-HG}
computes a 3-ruling set in $O((\log\log n)^3)$ rounds.
\end{theorem}
\begin{proof}
An MIS on $G[M_1\cup M_2]$ is a 3-ruling set for vertices that
join $M_1 \cup M_2 \cup W$ in the $i$th iteration of the algorithm as shown in
the proof of Theorem \ref{theorem:treesHighGirth}. In the rest of the proof, we only
concentrate on the runtime of Algorithm \textsc{RulingSet-HG} on graphs of
arboricity $a$.

The graph $H[M_j]$ for $j=1,2$ as defined in Lemma \ref{lemma:Mdegree} has an
arboricity of $a$ and poly-logarithmic degree. Hence, an MIS of $H[M_j]$ can be computed in
$O(\sqrt{\log n}\log\log n + \log^{3/4} n)$ rounds using 
\cite[Theorem 6.4]{BEPS12}. Since there
are $O(\log\log \Delta)$ iterations, the overall running time is $O(\sqrt{\log
n}\cdot (\log\log n)^2 + \log^{3/4} n\log\log n)$.

For small $a$, we can compute an MIS of $H[M_j]$, $j=1,2$ in time $O(\log
\Delta(H[M_j])\cdot (\log \Delta(H[M_j]) + \frac{\log\log n}{\log\log\log n}))$
rounds according to \cite[Theorem 6.4]{BEPS12}. Using this result with
$\Delta(H[M_j]) = O(\log n)$ for $j=1,2$, yields the theorem.
\end{proof}

\ignore {
\section{Removing Assumptions Global Knowledge}
Our algorithms in Section \ref{section:general}--\ref{section:arboricity}
require nodes to know the values of  $n$ and $\Delta$ to work correctly. In this
section, we argue that with small modifications to our algorithms, nodes need to
know only their local neighbors. 

To remove the knowledge of $\Delta$, we proceed as follows. Each node starts
several executions of the Algorithm {\sc RulingSet} in parallel. The $i$th
execution uses a value of $2^i$ as the guess for $\Delta$. Each node also
transmits to its neighbors an array $M$ where the $i$th element of the $M$
indicates whether this node is in $M$ in the $i$th execution or not. A value of
1 at index $i$ of $M$ indicates that node $v$ is in $M$ in the $i$th concurrent
execution. In this fashion, the message size that has to be transmitted by each
node is still in $O(\log n)$. 
}

\section{Conclusions}
\label{sec:conclusions}

Our work is the first positive evidence that $O(1)$-ruling sets can be
constructed much more quickly than an MIS and in sub-logarithmic rounds
even on general graphs. A major open question that our work raises is
the possibility of quickly extending an $O(1)$-ruling set to an MIS.
Another direction worth exploring is the application of our
sparsification technique to design sub-logarithmic time distributed
approximation algorithms.

\bibliographystyle{plain}
\bibliography{distComp}

\begin{thebibliography}{10}

\bibitem{AlonBabaiItai}
Noga Alon, L{\'a}szl{\'o} Babai, and Alon Itai.
\newblock A fast and simple randomized parallel algorithm for the maximal
  independent set problem.
\newblock {\em J. Algorithms}, 7(4):567--583, 1986.

\bibitem{BE08}
L.~Barenboim and M.~Elkin.
\newblock Sublogarithmic distributed {MIS} algorithm for sparse graphs using
  nash-williams decomposition.
\newblock In {\em ACM Symp. on Principles of Distributed Computing (PODC)},
  pages 25--34, 2008.

\bibitem{BE09}
L.~Barenboim and M.~Elkin.
\newblock Distributed ({$\delta+1$})-coloring in linear (in {$\delta$}) time.
\newblock In {\em STOC}, pages 111--120, 2009.

\bibitem{BE10}
L.~Barenboim and M.~Elkin.
\newblock Deterministic distributed vertex coloring in polylogarithmic time.
\newblock In {\em ACM Symp. on Principles of Distributed Computing (PODC)},
  pages 410--419, 2010.

\bibitem{BEPS12}
Leonid Barenboim, Michael Elkin, Seth Pettie, and Johannes Schneider.
\newblock Fast distributed algorithms for maximal matching and maximal
  independent set.
\newblock {\em CoRR}, abs/1202.1983, 2012.

\bibitem{BernsHegemanPemmaraju}
Andrew Berns, James Hegeman, and Sriram~V. Pemmaraju.
\newblock Super-fast distributed algorithms for metric facility location.
\newblock In {\em ICALP (2)}, pages 428--439, 2012.

\bibitem{GfellerPODC07}
Beat Gfeller and Elias Vicari.
\newblock A randomized distributed algorithm for the maximal independent set
  problem in growth-bounded graphs.
\newblock In {\em PODC '07: Proceedings of the twenty-sixth annual ACM
  symposium on Principles of distributed computing}, pages 53--60, 2007.

\bibitem{KOSS06}
K.~Kothapalli, C.~Scheideler, M.~Onus, and C.~Schindelhauer.
\newblock Distributed coloring in {$O(\sqrt{\log n})$} bit rounds.
\newblock In {\em In International Parallel and Distributed Processing
  Symposium, (IPDPS)}, 2006.

\bibitem{KMNWDISC05}
F.~Kuhn, T.~Moscibroda, T.~Nieberg, and R.~Wattenhofer.
\newblock Fast deterministic distributed maximal independent set computation in
  growth-bounded graphs.
\newblock In {\em in Proc. of Distribtued Computing}, pages 273--287, 2008.

\bibitem{KuhnMoscibrodaWattenhoferMAIN}
Fabian Kuhn, Thomas Moscibroda, and Roger Wattenhofer.
\newblock Local computation: Lower and upper bounds.
\newblock {\em CoRR}, abs/1011.5470, 2010.

\bibitem{LenzenWattenhofer}
Christoph Lenzen and Roger Wattenhofer.
\newblock Mis on trees.
\newblock In {\em Proceedings of the 30th annual ACM SIGACT-SIGOPS symposium on
  Principles of distributed computing}, PODC '11, pages 41--48, New York, NY,
  USA, 2011. ACM.

\bibitem{LubySICOMP86}
M.~Luby.
\newblock A simple parallel algorithm for the maximal independent set.
\newblock {\em SIAM Journal on Computing}, 15:1036--1053, 1986.

\bibitem{Metivieretal}
Y.~M\'{e}tivier, J.M. Robson, N.~Saheb-Djahromi, and A.~Zemmari.
\newblock An optimal bit complexity randomised distributed mis algorithm.
\newblock In {\em Proceedings of the 16th International Colloquium on
  Structural Information and Communication Complexity (SIROCCO)}, pages
  323--337, 2009.

\bibitem{NW64}
C.~Nash-Williams.
\newblock Decompositions of finite graphs into forests.
\newblock {\em J. London Math}, 39(12), 1964.

\bibitem{WattenhoferSchneiderPODC2008}
Johannes Schneider and Roger Wattenhofer.
\newblock A log-star distributed maximal independent set algorithm for
  growth-bounded graphs.
\newblock In {\em PODC}, pages 35--44, 2008.

\end{thebibliography}
\end{document}